\documentclass[11pt]{article}
\usepackage[english]{babel}
\usepackage[utf8]{inputenc}
\usepackage{amsmath, amsfonts, amsthm, amssymb, amscd,enumerate}
\newtheorem{theorem}{Theorem}[section]
\newtheorem{proposition}{Proposition}[section]

\numberwithin{equation}{section}
\newcommand{\Bf}[1]{{\bf{#1}}}

\renewcommand{\Pr}{\mathop{\rm Pr}}
\newcommand{\Ra}{\mathop{\rm Ra}}

\newlength{\captionwidth}
\newcommand{\uu}{\mathbf{u}}
\usepackage{graphicx}

\begin{document}

\title{Notes on Symmetry in Convective Flows}
\date{}
\author{Antonino De Martino, Arianna Passerini}
\maketitle

\begin{abstract}
In most fluid dynamics problems, the governing equations are  nonlinear because of the presence of convective terms. Nevertheless, existence of solutions can be shown by direct sum provided one identifies, in the relevant Banach space of solutions, particular subspaces
which are invariant under time evolution. As an example, we consider classical convection problems and show how symmetry arguments can help in identifying  such subspaces,  as well as the role it can play in some stability issues.

\end{abstract}

\section{Introduction}
Normal modes, decoupling of  unknowns and symmetries are main tools to achieve  solutions for systems of differential equations arising in mathematical physics. Linear systems, even coupled ones, can often be decoupled by simply projecting the data into suitable subspaces. In case of nonlinear PDE's this procedure  is, in general, not working, and it is not possible to split the solution into two components satisfying a {decoupled} system, even though there are  cases where this  might indeed occur. A most notable example is given by  the Navier-Stokes equations in a sufficiently smooth domain $\Omega \subseteq \mathbb{R}^n$ ($n=2,3$):

\begin{equation}
\label{NL1}
\begin{cases}
\nabla\cdot\mathbf{v}=0\\
\left(\frac{\partial
	\mathbf{v}}{\partial t}+\mathbf{v}\cdot\nabla\mathbf{v}\right)- \nu\Delta \mathbf{v}=- \nabla p \\
\end{cases}
\end{equation}
where
$\mathbf{v}(\mathbf{x},t)$ and $p(\mathbf{x},t)$ are unknown pressure and velocity fields and $\nu$ is the constant kinematic viscosity.
Then, the well-known Helmholtz-Weyl decomposition theorem allows one to solve \eqref{NL1} first for $\mathbf{v}$ and then recovering $p$ as the solution to a suitable Neumann problem for the Laplace operator \cite{Sohr}.

It must be added that the tools mentioned above, typically, also lead to find some physically significant exact solutions that may constitute the basic state for  stability analysis. Well-known examples are Couette and  Poiseuille flow, Asymptotic Suction Profile etc.,\footnote{For a vast collection of exact solutions to the Navier-Stokes equations, we refer to \cite{Ber}, \cite{Wang}.} and, in a non-isothermal situations, B\'enard purely conducting steady-state. Concerning the latter, we recall that the B\'enard problem in the Oberberbeck-Boussinesq (O-B) approximation regards the study of the following set of equations (in dimensionless form)
\begin{equation}
\label{L1}
\begin{cases}
\nabla \cdot \mathbf{v}=0\\
\frac{1}{\Pr} \biggl( \frac{\partial  \mathbf{v}}{\partial t}+  \mathbf{v} \cdot \nabla  \mathbf{v} \biggl)- \Delta \mathbf{v}=- \nabla\Pi+\Ra  \tau \mathbf{e}_3,\\
\frac{\partial \tau}{\partial t}+  \mathbf{v} \cdot \nabla \tau- \Delta \tau=  \mathbf{v}  \cdot \mathbf{e}_3.
\end{cases}
\end{equation}
where $\tau$ is  temperature field, $\Bf e_3$ is a unit upright vector parallel to the gravity, and $\Pr$ and $\Ra$ are Prandtl and Rayleigh number, respectively.

\par
In the first part of this note (Section 2) we show  by a general symmetry argument (Proposition 2.1) the existence of a particular subclass of solutions to \eqref{NL1} and its non-isothermal counterpart in the Boussinesq approximation. Of course, for this argument to work we need that the flow domain as well as the driving mechanism satisfy suitable symmetry conditions. We prove that  these solutions belong to a subspace, $\mathcal S$ (say), of the underlying Banach space that remains invariant under the relevant dynamics and  is stable to every perturbation with data in $\mathcal S$. As a consequence the {\em generic} solution can be split into two components: one, ${\bf u}_0$,  in $\mathcal S$, and the other, ${\bf u}$, in its complement. Thus, it is  the latter that produces possible instability and symmetry breaking of the bifurcating (stable) solution, due to the contribution of the convective (nonlinear) term given by ${\bf u}\cdot\nabla {\bf u}_0+{\bf u}_0\cdot\nabla {\bf u}.$
\par
The second part  (Section 3) concerns the instability of a thermal flow between two horizontal coaxial cylinders kept at different temperatures, with $\delta T$ (say) being their difference. As shown in \cite{PFRT}, for arbitrary values of  physical and geometric  parameters, there exists a steady-state flow, say {\sf s}$_0$, which, unlike the classical B\'enard problem between horizontal planes, has a nontrivial velocity field. Moreover, ${\sf s}_0$ is symmetric around the vertical diameter of the common cross-section of the cylinders. As expected on mathematical ground and suggested by numerical and experimental tests, this solution becomes non-unique and unstable if $|\delta T|$ is sufficiently large, all other parameters being fixed. However, the interesting question not entirely yet clarified is whether bifurcation occurs in a steady-state or time-periodic fashion and whether symmetry is preserved in the transition. We then show (Proposition 3.1) that symmetry breaking can occur only in a range of physical parameters which makes the smallest eigenvalue of a suitable eigenvalue problem strictly greater then 1.

\section{Continuous symmetries}
We give sufficient conditions for the occurrence of decoupling. These conditions are based on symmetries considerations.
More precisely, we  are interested in how symmetry properties of the domain (and of the driving mechanism), may naturally lead to find exact  solutions, or to prove their existence in subspaces possessing that particular symmetry.  In order to reach this goal, we briefly recall, for completeness,  some basic concepts of differential geometry.

Given a $C^2$-differentiable scalar function $f$ defined in an open set of $\mathbb{R}^{n}$, one can define a $(n-1)$-dimensional manifold whose points verify $f(\mathbf{x})=\lambda$. They can be identified by $n-1$ local Lagrangian coordinates in any open subset where the implicit function theorem can be applied. In particular, an alternative set of coordinates, the \textit{proper coordinates}, can be found for $\mathbb{R}^{n}$ such that the sub-manifold is defined by rquiring the vanishing of one. A unique set of Lagrangian coordinates (global atlas) cannot be found, in general, to describe the manifold as a whole. Nevertheless differential calculus can be fully defined on it.

If we allow $\lambda$ to vary, we can construct a partition of $\mathbb{R}^n$ by disjoint surfaces named \textit{leaves}, which are sub-manifolds of $\mathbb{R}^n$. Next, let us define an equivalence relation among points
by the requirement \textit{of being on the same surface}. Then, the set of the equivalence classes is a foliation and $\lambda$ is a coordinate for this quotient space.
Now, we are in position to give the first result.

\begin{proposition}
Consider either the Navier-Stokes or Oberbeck-Boussinesq (O-B)  systems in a domain $\Omega\subseteq\mathbb{R}^n$ which is a subset of a $C^2-$differentiable foliation.
Then there exist solutions which are constant on the leaves for all $t$ and the velocity is tangent to the leaves.
\end{proposition}

\begin{proof}
If each leaf is a sub-manifold $\lambda$ is constant on the leaves, then one sees that any \textit{vector field} \footnote{Here, by \textit{vector fields} we denote the differential operators in the tangent space of the manifold.} can be written as direct sum of fields which are respectively tangent to the leaves and to the quotient. Let us consider the class of scalar which are constant on the leaves, and let us compute their transport terms. If the velocity is taken tangent to the leaves, then such convective terms are nothing else than the directional derivative tangent to the leaf of functions which are constant on the leaves. Therefore, these terms vanish identically and the system becomes linear, provided the components of the unknown fields depend only on the quotient coordinate. Since we are now solving a linear equation the subspace so defined (the quotient is itself locally a subspace) is invariant under time evolution.
\end{proof}

As a first, elementary example, consider the Poiseuille flow in an infinite pipe $\Omega$ with circular (for simplicity) cross-section of radius $R$. In this case, if $r$ denotes the distance of a point in a cross-section to the axis of the cylinder, the surfaces defined by $r=c$, $c\in \mathbb R_+$, are the leaves. More precisely, if the axis of the cylinder is taken to be the $z$-axis of a cylindrical coordinate system $(r,\varphi,z)$
the corresponding (exact) solutions to \eqref{NL1} lie in  the vector subspace defined by the velocity field being of the form
\begin{equation}
\label{new}
(v^r,v^{\varphi},v^z)\equiv(0,0,v(r,t))
\end{equation}
In this case, the driving mechanism is an axial pressure gradient $C(t)=p(z+1,t)- p(z,t)$ depending on time $t$, at most.
Such a simple class of solutions is obtained as a consequence of the invariance of the domain by arbitrary translations along the axis and arbitrary rotations around it.
In the case when $C$ does not depend on  $t$, we have the well-known
 parabolic profile:
\begin{equation}
\label{EEQ0}
v(r)=\frac{C}{4\mu}(R^2-r^2)
\end{equation}
It is noteworthy that Poiseuille flows are, in fact, solutions to the Stokes problem as well, the latter being obtained by disregarding the nonlinear term in
\eqref{NL1}.
For any value of the physical parameter $C$,  solutions (steady or unsteady) in the class \eqref{new} can be shown to be  stable with respect to perturbations belonging to the  same class, in that it is a solution to the heat equation --as noticed in \cite{GR}-- and hence decays exponentially fast in time.
However, in order to study the evolution of a general 3-D perturbations in the stability analysis of steady Poiseuille flow in $(L^2(\Omega))^2$, one has to to control the trilinear form (with $\Bf v_0=v(r)\,\Bf e_3$):\footnote{$\langle\cdot,\cdot\rangle$ denotes the scalar product in $L^2$ and $\|\cdot\|_q$ the norm in $L^q$.}
\begin{equation}
\label{tril1}
\langle\mathbf{v}\cdot\nabla\mathbf{v}_0,\mathbf{v}\rangle\, \leq \sup_{\mathbf{r}\in (0,R)}|\nabla {v}(r)|\, \|\mathbf{v}\|^2_2\,,
\end{equation}
which imposes a restriction on the magnitude of $C$; see, for instance, \cite{P}.
\smallskip\par
Our next example is furnished by the
B\'{e}nard problem, which can be studied in 2-D to describe the so-called convective rolls. In such a case
the leaves are parallel, horizontal planes.
As is well known, the physical setting consists in a layer of motionless viscous fluid confined between two infinite plates at distance $h$, kept at constant and different temperatures with the lower plate being hotter than the upper one. The unperturbed state is then given by $\Bf v\equiv\Bf 0$ and a constant temperature gradient $\delta T/h$ throughout the layer, where $\delta T$ is the temperature difference between the plates. When $|\delta T|$ exceeds a certain critical valuer, convective motions occur in the shape of periodic  rolls; see, e.g., \cite{FM,GM}. In the Oberbeck-Boussinesq (O-B) approximation, the generic (nondimensional) perturbation $(\Bf v,\tau)$ to velocity and temperature fields must satisfy \eqref{L1}

In order to describe the periodic pattern of convective rolls, one imposes periodicity conditions in the horizontal direction, so that the relevant region of flow, becomes a ``periodicity cell" given by
$$ \Omega_0= \{(x,z) \in \mathbb{R}^2|\, \, x \in (0,1), z \in (0,1)\}.$$
where $z$ is the vertical coordinate, while, for simplicity, we assumed that the period in the horizontal direction is 1.
The boundary conditions for $\mathbf{v}$ are the traditional stress free conditions. So
\begin{equation}
\label{C41}
\tau(x,0,t)= \tau(x,1,t)=0\, ,
\end{equation}
\begin{equation}
\label{C2}
v^{z}(x,0,t)=v^{z}(x,1,t)=0\, ,
\end{equation}
\begin{equation}
\label{C3}
v_{z}^{x}(x,0,t)=v_{z}^{x}(x,1,t)=0\, ,
\end{equation}
where the subscript denotes partial derivative with respect to the indicated variable.
Boundary condition for $\Pi$ are of the Neumann type
\begin{equation}
\label{C4}
\Pi_{z}(x,0,t)=\Pi_{z}(x,1,t)=0\,,
\end{equation}
and, finally,
the initial conditions are
\begin{equation}
\label{C1}
( \mathbf{v}(x,z,0), \tau(x,z,0) )=(\mathbf{v}_{0}(x,z), \tau_{0}(x,z)).
\end{equation}

The existence of regular solutions for this system is well-known \cite[Section 3.5]{Tem}. Our objective is to prove that the solution can be obtained as direct sum of fields in certain subspaces, which are defining next.
Precisely, set
$$ \mathbb{S}:= \{ (\Pi, \mathbf{v}, \tau) \in C^{\infty}: \, \, \Pi_x= \mathbf{v}_x= \tau_x=0\, \, \, \hbox{in}\, \Omega; \, \, \Pi_z= v^{z}=v^{x}_z=\tau=0\,  \, \, \hbox{in} \, \,  z=0,1 \},$$
and
$$ \mathbb{F}:= \{(\Pi, \mathbf{v}, \tau) \in C^{\infty} \cap \mathbb{S}^{c}: \, \,  \hbox{periodic in}\, x\, ; \, \Pi_z= v^{z}=v^{x}_z=\tau =0\, \, \, \hbox{in} \, \, \, z=0,1 \}.$$
Notice that, as a consequence of the divergence free condition, one has $ \mathbf{v}= v(z) \mathbf{e}_1$ in $ \mathbb{S}$, with $\Bf e_1$ unit vector in the $x$-direction, whereas, in view of  the $x$-periodicity,  in $ \mathbb{F}$ one necessarily has mean value zero for all fields: $ \langle \Pi \rangle=\langle \mathbf{v} \rangle= \langle \tau \rangle=0$, where $\langle\cdot\rangle$ indicates average over $\Omega_0$.

Actually, for $(m,n) \in \mathbb{N}_{0} \times \mathbb{N}_{0}$, the relevant fields are found  as members of the  closure in $L^2$ of a complete set, $\mathcal B$,  as follows from \cite{CP}. For the  pressure field $ \Pi$, one chooses $\mathcal B$ as
\begin{equation}
\label{I1}
\phi_{mn}^{i}(x,z)= \begin{cases}
\cos( 2 \pi mx) \cos( \pi nz) \qquad \hbox{if} \quad i=1,\\
\sin( 2 \pi mx) \cos( \pi nz) \qquad \hbox{if} \quad i=-1,
\end{cases}
\end{equation}
while for the  temperature field $ \tau$ and the \emph{stream-function} $\Phi$ associated to $ \mathbf{v}$ by
$$\Phi_{x}=v^{z} \qquad \Phi_{z}=-v^{x}\, ,$$
$\mathcal B$ is given by
\begin{equation}
\label{I2}
\xi_{mn}^{i}(x,z)=\begin{cases}
\cos( 2 \pi mx) \sin( \pi nz) \qquad \hbox{if}  \quad i=1,\\
\sin( 2 \pi mx) \sin( \pi nz) \qquad \hbox{if}  \quad i=-1.
\end{cases}
\end{equation}
About the subscripts: the pair $(0,n)$ refers to $\mathbb{S}$, while the pair $(m,n) \in \mathbb{N} \times \mathbb{N}_{0},$ refers to $ \mathbb{F}.$
Notice that $\Pi$ can always be chosen with zero mean value, since it is defined up to a constant, while, of course, this is not true for the whole solution $ (\Pi, \mathbf{v}, \tau)$.
In particular, it is well known that $\Pi$ is obtained by solving  a Neumann problem associated to the equation
$$ \Delta \Pi= -\frac{1}{\Pr}\nabla \cdot ( \mathbf{v} \cdot \nabla \mathbf{v})+ \Ra \tau_{z},$$
which is indeed uniquely solvable since the right hand side vanishes because of the boundary conditions imposed on $\Bf v$ and $\tau$.

Let us now consider the reduced bases:
$$ \widetilde{\mathcal{B}}_{N}= \{ \phi_{nm}^{i}: (m,n) \in \mathbb{N} \times \mathbb{N}_{0} \quad \hbox{for} \quad i=1,2 \}.$$
$$ \widetilde{\mathcal{B}}_{D}= \{ \xi_{nm}^{i}: (m,n) \in \mathbb{N} \times \mathbb{N}_{0}\, \, \,  \hbox{or}\, \, \,  n=2k \, \, \,  \hbox{for} \, \, \,  i=1,2 \}.$$
The spaces of functions with vanishing mean value are then defined as follows. $ \widetilde{W}^{k,2}_{N}( \Omega_{0})$, $k=0,1,2$ is the closure of finite combinations of elements of $\widetilde{\mathcal{B}}_{N}$ in the Sobolev space $W^{k,2}(\Omega_{0})$.
Also, we denote by\footnote{Here, $N$ stands for Neumann and $D$ for Dirichlet.} $ (\widetilde{W}_{D}^{k,2}(\Omega_{0}))^2$ the space generated by $\tilde{\mathcal B}_D$ endowed with the norm of $W^{k,2}(\Omega_{0})$, for $k=0,1,2$. Analogously, we denote by $\widetilde{\mathcal{H}}_{D}(\Omega_{0})$ the space of vectorial divergence free functions generated by \eqref{I1} through the \emph{stream-function} , and endowed with the $L^{2}(\Omega_{0})$-norm. The space $(\widetilde{\mathcal{W}}_{D}^{k,2}(\Omega_{0}))^2$, $k=1,2$, is accordingly defined by the $W^{k,2}( \Omega_{0})$-norm.

As previously recalled, the problem \eqref{L1}, \eqref{C41}--\eqref{C1} possesses a unique, global regular solution in the sense of Ladyzhenskaya \cite[Section 3.5]{Tem}.
In particular, if $(\mathbf{v}_0,\tau_0)\in (\mathcal{W}^{1,2}_{D}(\Omega_{0}))^2\times W^{1,2}_{D}(\Omega_{0})$ then the solution satisfies
$$ \Pi \in L^{2} (0,T; W^{1,2}_{N}( \Omega_{0})),$$
$$ \mathbf{v} \in [L^{\infty}( 0,T; \mathcal{W}^{1,2}_{D}(\Omega_{0}) ]^2 \cap [L^{2}(0,T; \mathcal{W}^{2,2}_{D}( \Omega_{0}) ]^2,$$
$$ \tau \in L^{\infty} (0,T; W^{1,2}_{D}(\Omega_{0}) \cap L^{2}(0,T); W^{2,2}_{D}( \Omega_{0}))$$
We shall now introduce a suitable decomposition of the solution that allows us to calculate the mean value of all the quantities without computing the whole solution.

The basic solution $(\Pi, \mathbf{v}, \tau)\equiv(0, c\,\mathbf{e}_1, 0)$, $c$ a constant, satisfies system \eqref{L1} and the stress free boundary conditions \eqref{C41}-\eqref{C4}. Since $c$ is an  arbitrary real number,  this class of solutions corresponds to the strict Galileian invariance of the equations (see \cite{CP}) and we can call it null solution.

Now, we shall show that the space
 $\mathbb{S}$ is  an invariant, stable subspace of the solution space defined above. We begin to observe that
elements of $\mathbb S$ are triple of the type $(\Pi(z,t), a(z,t)\mathbf{i}, \tau(z,t))$.
Moreover, even if the initial data are not given in $\mathbb S$, the projection of the full system onto $\mathbb{S}$ does not depend on the solution in its complement, whereas the solution in the complement depends on the solution in $\mathbb{S}$.
By looking in general for solutions to \eqref{L1} of the  form
\begin{equation}\label{classe}(\mathcal{G}(z,t), \mathcal{A}(z,t)\,\mathbf{e}_1, \mathcal{T}(z,t))\,,
\end{equation}
we see that they should satisfy the linear system
\begin{equation}\label{stable}
\begin{cases}
\frac{\partial^{2} \mathcal{G}}{\partial z^{2}}= \Ra \frac{\partial \mathcal{T}}{\partial z}\\
\frac{1}{\Pr} \frac{\partial \mathcal{A}}{\partial t}- \frac{d^{2} \mathcal{A}}{dz^{2}}=0  \\
\frac{\partial  \mathcal{T}}{\partial t}- \frac{\partial^{2} \mathcal{T}}{\partial z^{2}}=0.
\end{cases}
\end{equation}
We have the following result of simple proof.
\vspace{0.5cm}

\begin{proposition}
	For all $Pr$ and $Ra$, problem \eqref{L1} with initial conditions
	$$ \begin{cases}
	v^{x}(x,z,0)=f(z) \\
	v^{z}(x,z,0)=0\\
	\tau(x,z,0)=g(z).
	\end{cases} $$
	with boundary conditions \eqref{C2}-\eqref{C3} and $f,g \in L^1(0,1)$ has a  solution $(\Bf v,\tau)$ in the space $C^{\infty}([\eta,\infty)\times \Omega_{0})$, for all $\eta>0$, which is also unique in the class \eqref{classe}.
\end{proposition}
\begin{proof}
We look for solutions of the form $\Bf v= \mathcal{A}(z,t) \mathbf{e}_1, \tau=\mathcal{T}(z,t))$.  Then, as already noticed, $\mathcal A$ and $\mathcal T$ must satisfy \eqref{stable}$_{2,3}$ with a corresponding pressure $\mathcal G$ given in \eqref{stable}$_1$.
We then find
$$\begin{array}{ll}\medskip \displaystyle{\mathcal{A}= 2 \sum_{n=1}^\infty  \biggl( \int_{0}^1 f(s)\cos(n \pi s) \, ds \biggl) \cos(n \pi z) e^{-n^{2}  \pi^{2} t}
}\\ \displaystyle{
\mathcal{T}=2 \sum_{n=1}^\infty  \biggl( \int_{0}^1 g(s)\sin (n \pi s) \, ds \biggl) \sin(n \pi z) e^{-n^{2}  \pi^{2} t},}\end{array}
$$
which gives \cite{CP}
$$
\mathcal{G} =\mathcal{G}_0+2\Ra \sum_{n=1}^{\infty} \frac{ e^{-n^{2}  \pi^{2} t}}{n\pi} \big(\int_{0}^1 g(s)\sin (n \pi s) \, ds \big)(1-\cos n\pi z)\, .
$$
\end{proof}
In order to see how the projection on the space $\mathbb S$ may affect a generic flow, we need a further estimate for $\mathcal{A}\in L^{\infty}((0,T); W^{1,2}_{N}(\Omega_{0}))$. So, we choose $f(z), f'(z), f''(z) \in L^{2}( \Omega_{0})$ and differentiate the above solution with respect to $z$, we get
\begin{eqnarray*}
\sup_{t \in [0,T]} \| \mathcal{A}_{z} \|_{2} &\leq& \sqrt{\sum_{n=0}^{\infty} n^{2} \biggl( \int_{0}^{1} f(s) \sqrt{2} \cos(n \pi s) \, ds \biggl)^{2} e^{-2 n^{2} \pi^{2} t}} \\
&\leq& \sqrt{\sum_{n=0}^{\infty} n^{2} \biggl( \int_{0}^{1} f(s) \sqrt{2} \cos(n \pi s) \, ds \biggl)^{2}},
\end{eqnarray*}
that is,
\begin{equation}
\label{costanti}
| \mathcal{A}|_{\infty, 1,2} \leq \| f' \|_{2}.
\end{equation}
In a similar way
\begin{equation}
\label{costanti1}
| \mathcal{A}|_{\infty, 2,2} \leq \| f'' \|_{2}.
\end{equation}

If the projection of the initial conditions in $\mathbb{F}$ is different from zero, the generic solution $(\Bf v,\tau)$ to \eqref{L1} can be split as $ \mathbf{v}= \mathbf{u}+ \mathcal{A}\mathbf{e}_1$, $\tau=\sigma+\mathcal T$, with $\Bf u,\tau \in \mathbb F$. Since, by a straightforward calculation,
$$\nabla \cdot (\mathbf{v} \cdot \nabla \mathbf{v})= \nabla \cdot [\mathbf{u} \cdot \nabla \mathbf{u}+\mathcal{A} \mathbf{u}_{x}]+ \mathcal{A}_{z}u_{x}^{z}\,,$$
we deduce that the triple $(\Bf u,\sigma, \mathcal P=\Pi-\mathcal G)$ must satisfy the following set of equations with $\Bf A=\mathcal A\,\Bf e_1$
\begin{equation}
\label{L1bis}
\begin{cases}
\Delta \mathcal{P}= - \frac{1}{\Pr}[ \nabla \cdot ( \mathbf{u} \cdot \nabla \mathbf{u})+ \nabla \cdot ( \mathcal{\mathbf{A}} \cdot  \nabla \mathbf{u}+ \mathbf{u}  \cdot \nabla \mathcal{\mathbf{A}}  )]+ \Ra \sigma_{z},\\
\frac{1}{\Pr} \biggl( \frac{\partial \mathbf{u}}{\partial t}+ \mathbf{u} \cdot \nabla \mathbf{u}+ \mathcal{\mathbf{A}} \cdot  \nabla \mathbf{u}+ \mathbf{u}  \cdot \nabla \mathcal{\mathbf{A}} \biggl)- \Delta \mathbf{u}= - \nabla \Sigma +\Ra  \sigma \mathbf{e}_3,\\
\frac{\partial \sigma}{\partial t}+\mathbf{u}\cdot \nabla \sigma+\mathbf{A} \cdot  \nabla \sigma- \Delta \sigma= u^{z}.
\end{cases}
\end{equation}
where  $$\Sigma=\mathcal{P}+ \mathcal{G}- \Ra \int_{0}^{z} \mathcal{T}(\mathit{z},t) \, d \mathit{z}.$$
The influence of a non-zero  component in $\mathbb S$ on the generic flow  can be seen by performing energy estimates on the ``perturbed" system \eqref{L1bis}.\footnote{For generalized O-B Bénard system energy estimate see \cite{DP}.} Thus, if we multiply both sides of \eqref{L1bis}$_{2,3}$ by $\Bf u$ and $\sigma$, respectively and integrate by parts, we get
$$
\mbox{$\frac12$}\frac{dE}{dt}+\|\nabla\Bf u\|_2^2+\|\nabla\sigma\|_2^2=\Ra\langle \sigma,u^z\rangle-\frac1\Pr\langle \Bf u\cdot\nabla\Bf A,\Bf u\rangle\,,
$$
where $E(t):=  \frac1\Pr{\| \uu \|_{2}^{2}}+ \Ra \| \sigma \|_{2}^{2}$. Observing that
$$ |\langle\uu \cdot \nabla \textbf{A}, \uu\rangle|\, \leq | \mathcal{A}|_{\infty,1,2} \| \uu \|_{2}^{2} \leq \| f' \|_{2} \| \uu \|_{2}^{2}.
$$
we deduce
$$ \mbox{$\frac12$}\frac{dE}{dt}+ \| \nabla \uu \|_{2}^{2}+ \Ra \| \nabla \sigma \|_{2}^{2} \leq C(\Ra, \Pr,  \| f' \|_{2}) E(t)\, ,$$
from which the boundedness of $E(t)$ follows.
In the same fashion,  we can prove analogous estimates for ``higher order" energies.
Nevertheless, there is a relevant feature of the decomposition that we would like to point out. In case we were interested only in the mean value of the physical quantities we are able to do it without computing the solution of \eqref{L1bis}.
In fact, since the elements of $\mathbb{F}$ have zero mean value, it follows immediately that all the averaged quantities can be found just by computing the part of the solution that lies in $ \mathbb{S}$. To this end it suffices to use just the solution obtained by prescribing as initial condition the mean value with respect to $x$ of the full initial condition. Precisely, we have the following result of immediate proof.
\begin{proposition} For all $Pr$ and $Ra$, the mean values of the solution of the full problem \eqref{L1} are equal to the averaged solutions of \eqref{stable} with initial conditions
	$$ \begin{cases}
	v_0^{x}=\int_0^1 v^{x}(x,z,0) dx \\
	v_0^{z}=0\\
	\tau_0=\int_0^1 \tau(x,z,0) dx.
	\end{cases} $$
\end{proposition}
\section{A discrete symmetry problem}
We are interested in the 2D-convective motion of a Navier-Stokes liquid between two horizontal coaxial cylinders with radii $R_{o}$ and $R_{i}\,(< R_{o})$, when  temperature distributions $\Theta_{i}$, on the inner jacket, and $\Theta_{o}$, on the outer jacket, are prescribed, with $\Theta_{o} < \Theta_{i}$ see \cite{PR,PF,PRT,DFPP,PFRT}. In such a case, one observes steady-state  flow,  no matter how small $\delta \Theta=\Theta_i-\Theta_o$.
The precise mathematical formulation
of the problem goes as follows. Denote by $(r,\varphi)$ a set of polar coordinates in the cross-sectional plane of both cylinders,  with the origin on the common axis.
In the first place, one needs a suitable lifting of the temperature distribution at the boundary, say $\Theta^{*}$,  which corresponds to the state of pure conduction:
$$ \Delta \Theta^{*}=0, \qquad \Theta^{*}(R_{i})=\Theta_{i} \qquad \Theta^{*}(R_{o})=\Theta_{o}.$$
We thus get
$$ \Theta^{*}=\Theta_{i}+ \frac{\Theta_{o}-\Theta_{i}}{\mathcal{B}}( \log r-\log R_{i}),$$
where
\begin{equation}\label{Bi}
\mathcal{B}:= \log \bigl( \frac{R_{o}}{R_{i}} \bigl)= \log \biggl(1+ \frac{2}{\mathcal{D}} \biggl), \quad \hbox{where} \quad \mathcal{D}:= \frac{2R_{i}}{R_{o}-R_{i}}.
\end{equation}
The (dimensionless) parameter $\mathcal D$  is related to the the ``geometry" of the domain. A ``large" value of $\mathcal D$ implies a narrow gap, while a ``small" one would mean the opposite.
The region of flow, in dimensionless form, is then given by
$$ \Omega_{\mathcal{D}}:= \{(r, \varphi) \in \mathbb{R}^{2}: r \in  ( \mathcal{D}/2,1+ \mathcal{D}/2 ) \}.$$
Thus, after introducing the  "deviatory-temperature"
$$ \tau:= \Theta-\Theta^{*}=\Theta-\Theta_{i}+ \frac{\Theta_{i}-\Theta_{o}}{\mathcal{B}}( \log r - \log R_{i}),$$
and defining the non-dimensional quantities
$$ \tilde{t}= \frac{\kappa}{(R_{o}-R_{i})^{2}}t, \quad \tilde{r}:= \frac{r}{R_{o}-R_{i}}, \quad \tilde{z}:= \frac{z}{R_{o}-R_{i}}, \quad \tilde{\tau}:= \frac{\tau}{\Theta_{i}-\Theta_{o}}, \quad \tilde{R_i}:=\mbox{$\frac12$}\mathcal D,\quad \tilde{R_o}:=1+\mbox{$\frac12$}\mathcal D \,,
$$
with $z$ vertical coordinate,
one can show that
the relevant governing equations (in non-dimensional form) are given by (tildes omitted)\footnote{In polar coordinates $ \uu \cdot \nabla \uu$ has the following components:
	$$ \textbf{u} \cdot \nabla \textbf{u}=
	\begin{pmatrix}
	\partial_{r} v^{r} \quad \frac{1}{r} (\partial_{\varphi}v^{r}-v^{\varphi})\\
	\partial_{r} v^{\varphi} \quad \frac{1}{r} (\partial_{\varphi}v^{\varphi}+v^{r})
	\end{pmatrix}
	\begin{pmatrix}
	v^{r}\\
	v^{\varphi}
	\end{pmatrix}
	=
	\begin{pmatrix}
	v^{r} \partial_{r} v^{r} + \frac{v^{\varphi}}{r}( \partial_{\varphi} v^{r}-v^{\varphi})\\
	v^{r}\partial_{r} v^{\varphi}+ \frac{v^{\varphi}}{r}(v^{r}+ \partial_{\varphi} v^{\varphi})
	\end{pmatrix}
	$$}

\begin{equation}
\label{T1}
\begin{cases}
\nabla \cdot \textbf{v}=0\\
\frac{1}{\Pr} \biggl( \frac{\partial \textbf{v}}{\partial t}+ \textbf{v} \cdot \nabla \textbf{v} \biggl)=- \nabla \Pi +\Delta \textbf{v}+ \frac{\Ra}{\mathcal{B}} \sin \varphi \textbf{e}_{r}+\Ra \tau \textbf{e}_{3}\\
\frac{\partial \tau}{\partial t}+ \textbf{v} \cdot \nabla \tau- \frac{v_{r}}{r \mathcal{B}}= \Delta \tau\,.
\end{cases}
\end{equation}
Here $\textbf{v}$ and $\Pi$ are velocity and (modified) pressure fields of the liquid, $v_{r}$ the radial component of $\textbf{v}$, while  Prandtl and Rayleigh numbers are given by
\begin{equation}
\Pr:= \frac{\nu}{ D}, \qquad \Ra:= \frac{\alpha g (\Theta_{i}-\Theta_{o})(R_{o}-R_{i})^{3}}{\nu D}\, ,
\end{equation}
where $\nu$ and $D$ are kinematic viscosity and thermal diffusivity coefficients, respectively,  $\alpha$ is the coefficient of volume expansion and $g$ the acceleration of gravity. Finally, $\textbf{e}_r$ and $\textbf{e}_3$ are unit vectors in the radial and upright vertical direction.
To (\ref{T1}) we append the boundary conditions
\begin{equation}
\label{T2}
\textbf{v}(t, R_{i}, \varphi)= \textbf{v}(t,R_{o}, \varphi)=\textbf{0} \qquad
\tau(t,R_{i}, \varphi)= \tau(t, R_{o}, \varphi)=0.
\end{equation}
Furthermore, all fields are of course periodic in $\varphi$.

It is important to emphasize that, since $\nabla\times (\sin\varphi\,\textbf{e}_r)\neq\textbf{0}$, unlike the classical B\'enard convection problem between two horizontal planes, in the problem at hand  the state of pure conduction, namely, the trivial solution $(\textbf{v}\equiv\textbf{0}, \tau\equiv0$), is {\em not} allowed.
Now, given the geometry of the flow region, one expects that such non-trivial flow possesses mirror symmetry property with respect to the vertical diameter of the annulus. Precisely, denoting by $(v_r,v_\varphi)$ the polar components of the velocity field, this  property translates in
the following \emph{even symmetry} condition  (see also \cite[Figure 2]{PFRT}):
\begin{equation}\nonumber
\label{1}
v^{r}(r, \varphi)=v^{r}(r, \pi- \varphi)\, ,
\end{equation}
\begin{equation}
\label{2}
v^{\varphi}(r, \varphi)=- v^{\varphi}(r, \pi- \varphi)\, ,
\end{equation}
\begin{equation}\nonumber
\tau(r, \varphi)=\tau(r, \pi- \varphi)\, .
\end{equation}
As shown in \cite{PFRT}, the set of equations in \eqref{T1} is invariant under the symmetry \eqref{2}.  However, though physically intuitive, the proof of existence of steady-state  solutions of the type (\ref{1}) for {\em arbitrary} values of the parameters is not so obvious. In principle, one could simply notice that once a solution $(\mathbf{v},\tau)$ is found, then after the change $\varphi \to \pi-\varphi$ another solution is automatically given by the right-hand sides in \eqref{2}. Thus, the found steady-state solution  is necessarily symmetric, provided one proves its uniqueness.
However, uniqueness holds only under certain restrictions on the size of the parameters and, in fact,  it is not even expected, as can be deduced by numerical simulations (see for instance \cite{Y}) showing multiple steady solutions in some regions of the dimensionless parameter space.
Therefore, existence of symmetric steady solutions without restrictions on the magnitude of the parameters should be possibly proved directly in the class of functions satisfying \eqref{2}. This was done in \cite{PFRT} by employing the classical Galerkin   method with a suitable base of functions in the class \eqref{2}. Precisely, setting
$$
\mathcal W^{1,2}(\Omega):=\{\textbf{v}\in [W_0^{1,2}(\Omega)]^2: \ \nabla\cdot\textbf{v}=0\}\,,
$$
with $W^{1,2}_0$ Sobolev space of functions vanishing at the boundary, in \cite{PFRT} the following existence result is proved.
\begin{theorem}\label{teo}
For any $\Pr$, $\Ra$ and $\mathcal{D}$, there exists at least one even symmetric steady-state solution $(\textbf{v}_0,\tau_0)\in \mathcal W^{1,2}(\Omega)\times W^{1,2}_0(\Omega)$ of system \eqref{T1} verifying the following estimates
\begin{equation} \label{bi}
\| \nabla \mathbf{v}_{0} \|_{2} \leq \Ra \frac{c}{\mathcal{B}} \quad \hbox{and} \quad \| \nabla \tau_{0} \|_{2} \leq \Ra \frac{c}{\sqrt[3]{\mathcal{B}^{2}}},
\end{equation}
where $c \leq c( \mathcal{D}^{*})$ for all $ \mathcal{D}< \mathcal{D}^{*}.$ Moreover, there is ${\rm Ra}_*>0$ such that this solution is unique for ${\rm Ra}<{\rm Ra}_*$.
\end{theorem}

Now, as mentioned earlier on, it is expected that for ``large" Ra the above solution is no longer unique and bifurcation (steady or time-periodic) may occur. This issue has been the object of both numerical and experimental tests, mostly, in the case when $\mathcal D\to 0$. However, their outcome seems to be at odds (see \cite{PFRT} and the reference therein). More precisely, numerical tests \cite{Y} suggest that ``exchange of stabilities" occurs, with the bifurcating branch being steady and with  {\em even symmetry}, whereas the experimental ones show  transition to a stable time-periodic solution.
The latter consists of an oscillatory motion of the streamlines around the axis of the cylinder, which implies {\em breaking of the even
 symmetry}. Our objective here is to furnish a contribution to this problem, by furnishing a distinctive lower bound at which thi transition may occur.

To this end, let $(\textbf u,{\sigma/\sqrt{{\rm Ra}}})$ be the generic perturbation to the steady-state motion $(\textbf v_0,\tau_0)$. From \eqref{T1} we thus get
\begin{equation}
\label{Tw}
\begin{cases}
\nabla \cdot \textbf{u}=0\\
\frac{1}{\Pr} \biggl( \frac{\partial \textbf{u}}{\partial t}+ \textbf{u} \cdot \nabla \textbf{u}+\textbf{u} \cdot \nabla \textbf{v}_0+\textbf{v}_0\cdot \nabla \textbf{u} \biggl)-\Delta\textbf{u}=- \nabla Q +\sqrt{{\rm Ra}}\, {\sigma} \textbf{e}_{3}\\
\frac{\partial {\sigma}}{\partial t}+ \textbf{u} \cdot \nabla  {\sigma}+\sqrt{{\rm Ra}}\,\textbf{u} \cdot \nabla \tau_0+\textbf{v}_0 \cdot \nabla {\sigma}-\Delta {\sigma}=\sqrt{{\rm Ra}}\,\displaystyle{\frac{u^r}{r \mathcal{B}}} \,.
\end{cases}
\end{equation}
Formally multiplying the second equation by $\textbf u$, the third by $\sigma$ and then integrating by parts over $\Omega$, we deduce
the ``perturbation energy equation''
\begin{equation}
\label{bif}
\frac{d}{dt}E+\frac{1}{\Pr}\langle\textbf{u}\cdot\nabla\textbf{v}_0,\textbf{u}\rangle+{\sqrt{\Ra}}\,\!\langle\textbf{u}\cdot\nabla\tau_0, \sigma\rangle+\|\nabla\textbf{u}\|^2+\|\nabla \sigma\|^2=\sqrt{\Ra}\langle u^z+\frac{u^r}{r\mathcal{B}},\sigma\rangle\,,
\end{equation}
where
$$
E=\frac1{2{\rm Pr}}\|\textbf u\|^2+\frac12\|\sigma\|^2\,,
$$
and, we recall, $\langle\cdot,\cdot\rangle$ and $\|\cdot\|_2$ denote scalar product and norm in $L^2(\Omega)$. Setting
\begin{equation}\begin{array}{cc}\medskip\medskip
\mathcal{F}(\mathbf{u},\sigma,\Pr,\Ra):=\displaystyle{\frac{\sqrt{\Ra}\langle u^z+\frac{u^r}{r\mathcal{B}},\sigma\rangle-\frac{1}{\Pr} \langle\textbf{u}\cdot\nabla\textbf{v}_0,\textbf{u}\rangle- {\sqrt{\Ra}}\,\!\langle\textbf{u}\cdot\nabla\tau_0, \sigma\rangle}{\|\nabla\mathbf{u}\|^2+\|\nabla \sigma\|^2}}\,,
\end{array}
\end{equation}
from \eqref{bif} we deduce
\begin{equation}\label{energy}
\frac{d}{dt}E=(\mathcal{F}(\mathbf{u},\sigma,\Pr,\Ra)-1)\left(\|\nabla\mathbf{u}\|^2+\|\nabla \sigma\|^2\right)\,.
\end{equation}
The following result holds.
\begin{proposition}
Let
\begin{equation}
M(\Pr,\Ra)=\displaystyle{\max_{(\mathbf{u},\sigma)\in\mathcal{W}^{1,2} \times W_0^{1,2}}\mathcal{F}(\mathbf{u},\sigma,\Pr,\Ra)}\,.
\label{max}
\end{equation}
Then $M$ exists for all values of $\Pr$ and $\Ra$.
\end{proposition}
{\em Proof.} Let us begin to show that the functional $\mathcal  F$ is bounded above. To this end, we recall the inequality of Poincar\'e:
$$
\|\Bf u\|_2\le \gamma_0\|\nabla\Bf u\|_2\,,\ \ \|\sigma\|_2
\le \gamma_0\|\nabla\sigma\|_2\,,
$$
and that of Ladyzhenskaya:
$$
\|\Bf u\|_4\le \gamma_0\|\nabla\Bf u\|_2\,,\ \ \|\sigma\|_4\le \gamma_0\|\nabla\sigma\|_2\,,
$$
where the constant $\gamma_0$ can be taken independent of $\mathcal D$ \cite{PFRT}.
Next, denoting by $\mathcal I=\mathcal I(\Bf u,\sigma,\Pr,\Ra)$ the numerator in the functional $\mathcal F$, by the latter and H\"older inequality we get
$$\begin{array}{rl}\medskip
|\mathcal I|&\!\!\!\le c(\mathcal B)\,\left[\sqrt{\Ra}\left (\|\nabla\Bf u\|_2^2+\|\nabla\sigma\|_2^2+\|\Bf u\|_4\|\sigma\|_4\|\nabla\tau_0\|_2\right)+\frac1\Pr\|\Bf u\|_4^2\|\nabla\Bf v_0\|_2\right]
\\
&\!\!\! \le c(\mathcal B,\Bf v_0,\tau_0) \left[\,\sqrt{\Ra}(\|\nabla\Bf u\|_2^2+\|\nabla\sigma\|_2^2)+\frac1\Pr\,\|\nabla\Bf u\|_2^2\right]\,,
\end{array}
$$
which proves the desired boundedness. We now rescale $\Bf u$ and $\Bf \sigma$ by setting
$$
\Bf w=\frac{\Bf u}{\|\nabla\Bf u\|_2^2+\|\nabla\sigma\|_2^2}\,,\ \ \zeta=\frac{\sigma}{\|\nabla\Bf u\|_2^2+\|\nabla\sigma\|_2^2}
$$
so that, denoting by $\{\Bf w_k,\zeta_k\}$ a maximizing sequence, we have
$$
\lim_{k\to\infty}\mathcal I(\Bf w_k,\zeta_k,\Pr,\Ra)=\ell\,,\ \ 1=\|\nabla\Bf w_k\|_2^2+\|\nabla\zeta_k\|_2^2:=D(\Bf w_k,\zeta_k)
$$
with
$$
\ell=\sup_{\footnotesize {\begin{array}{cc}(\mathbf{u},\sigma)\in\mathcal{W}^{1,2} \times W_0^{1,2}\\
D(\Bf w,\zeta)=1\end{array}}
}\mathcal{I}(\mathbf{w},\zeta,\Pr,\Ra)\,.
$$
Because of Poincar\'e inequality and the normalization condition, the sequence $\{\Bf w_k,\zeta_k\}$ is bounded in $\mathcal W^{1,2}(\Omega)\times W^{1,2}_0(\Omega)$ and, therefore,  there is $(\bar{\Bf w},\bar{\zeta})\in \mathcal W^{1,2}(\Omega)\times W^{1,2}_0(\Omega)$ such that
\begin{equation}\label{conv}\begin{array}{ll}\medskip
\{\Bf w_k,\zeta_k\}\to (\bar{\Bf w},\bar{\zeta})\,,\ \ \mbox{weakly  in $\mathcal W^{1,2}(\Omega)\times W_0^{1,2}(\Omega)$}\\
\{\Bf w_k,\zeta_k\}\to (\bar{\Bf w},\bar{\zeta})\,,\ \ \mbox{strongly  in $[L^q(\Omega)]^2\times L^q(\Omega)$\,,\ for all $q\in[1,\infty)$.}
\end{array}
\end{equation}
Employing these convergence properties, it is then simple to show that
\begin{equation}
\ell=\lim_{k\to\infty}\mathcal I(\Bf w_k,\zeta_k,\Pr,\Ra)=\mathcal I(\bar{\Bf w},\bar{\zeta},\Pr,\Ra)\,.\label{sept}
\end{equation}
Let us now prove that
\begin{equation}
D(\bar{\Bf w},\bar{\zeta}):=\|\nabla\bar{\Bf w}\|_2^2+\|\nabla\bar{\zeta}\|_2^2=1.
\label{n1}
\end{equation}
We shall treat the case $\ell>0$, because it is the only one relevant from a physical viewpoint. Since $D(\Bf w_k,\zeta_k)=1$, for all $k\in\mathbb N$, it follows that
\begin{equation}\label{D}
D(\bar{\Bf w},\bar{\zeta})\le1\,.
\end{equation}
Moreover, by definition of lower upper bound, we must have that there is $\bar{k}\in\mathbb N$ such that
$$
\mathcal I(\Bf w_k,\zeta_k,\Pr,\Ra)\ge \frac{I(\bar{\Bf w},\bar{\zeta},\Pr,\Ra)}{D(\bar{\Bf w},\bar{\zeta})}\,,\ \ \mbox{for all $k\ge\bar{k}$.}
$$
Thus, passing to the limit $k\to\infty$ and using \eqref{sept} we deduce
$$
\mathcal I(\bar{\Bf w},\bar{\zeta},\Pr,\Ra)\ge \frac{I(\bar{\Bf w},\bar{\zeta},\Pr,\Ra)}{D(\bar{\Bf w},\bar{\zeta})}
$$
which implies $D(\bar{\Bf w},\bar{\zeta})\ge1$. The latter
 in combination with \eqref{D} implies \eqref{n1}. We have thus shown that there is $(\bar{\Bf w},\bar{\zeta})\in \mathcal W^{1,2}(\Omega)\times W_0^{1,2}(\Omega)$ such that
$$
\ell=\frac{I(\bar{\Bf w},\bar{\zeta},\Pr,\Ra)}{D(\bar{\Bf w},\bar{\zeta})}\,,
$$
which completes the proof of the proposition.\hfill$\square$\smallskip\par
In view of the above proposition, from \eqref{energy} we infer that the perturbation energy $E=E(t)$ is a non-increasing function of time, and hence $(\textbf v_0,\tau_0)$ is stable, provided
\begin{equation}
\label{stab}
M(\Pr,\Ra)\le1\,.
\end{equation}
In particular --by taking $E(t)$ independent of time in \eqref{energy}--  condition \eqref{stab} implies that $(\textbf v_0,\tau_0)$ is the {\em only } steady-state solution to \eqref{Tw}. As a consequence, we deduce that {\em breaking of even symmetry may occur only if $M(\Pr,\Ra)>1$.}
\par
As is well known, the evaluation of $M$ is performed by solving the Euler-Lagrange equations associated to the maximum problem \eqref{max}. More precisely,  consider the following eigenvalue problem
\begin{equation}\label{eig}\begin{array}{cc}\medskip\smallskip\left\{\begin{array}{ll}\medskip
\nabla\cdot\Bf u=0\\
\medskip
\Delta{\bf u}-\nabla p=\lambda\,\left(\sigma\,\nabla\tau_0-\sigma{\bf e}_3-\sigma\,\frac{{\bf e_r}}{r\mathcal B}\right)+\frac1{\Pr}{\bf u}\cdot \Bf D(\Bf v_0)\\
\Delta\sigma=\lambda\left(\Bf u\cdot\nabla\tau_0-u^z-\frac{u^r}{r\mathcal B}\right)\,,
\end{array}\right.\\
\textbf{u}(R_{i}, \varphi)= \textbf{u}(R_{o}, \varphi)=\textbf{0} \qquad
\sigma(R_{i}, \varphi)= \sigma(R_{o}, \varphi)=0\,,
\end{array}
\end{equation}
where $\Bf D(\Bf v_0)=\frac12(\nabla\Bf v_0+\nabla\Bf v_0^\top)$.
From the general theory given in \cite{GaSt}, it follows that \eqref{eig} has a denumerable number of increasing eigenvalues clustering at infinity, once the relevant operators are suitably defined. Denoting by $\lambda_c$ the least eigenvalue, one can then show that $M=2/\lambda_c$ \cite[Lemma 4]{GaSt}. Moreover, an eigenfunction corresponding to $\lambda_c$, say $(\Bf u_c,\tau_c)$, has an important physical meaning, in that if $M>1$ (so we are outside the stability region according to energy theory) and we take in \eqref{Tw} $(\Bf u_c,\tau_c)$ as initial data, the perturbation energy $E$ is initially {\em strictly increasing}. In other words, $(\Bf u_c,\tau_c)$ is the ``most dangerous" perturbation for the onset of instability. The symmetry properties of this eigenfunction can be investigated at least in the limit of small $\mathcal D$ or, equivalently, large $\mathcal B$, according to the following argument.  Set $\varepsilon:=1/\mathcal B$ and assume that the steady-state solution  $(\Bf v_0(\varepsilon),\tau_0(\varepsilon))$ is a sufficiently smooth function of $\varepsilon$. From Theorem \ref{teo} it follows that $(\Bf v_0(0),\tau_0(0))=(\Bf 0,0)$. If we formally evaluate \eqref{eig} at $\varepsilon=0$ we get
\begin{equation}\label{eig0}\begin{array}{cc}\medskip\smallskip\left\{\begin{array}{ll}\medskip
\nabla\cdot\Bf u(0)=0\\
\medskip
\Delta{\bf u}(0)-\nabla p(0)=-\lambda(0)\,\sigma(0)\Bf e_3\\
\Delta\sigma(0)=-\lambda(0)u_z(0)\,,
\end{array}\right.\\

[\textbf{u}(0)](0)= [\textbf{u}(0)](1, \varphi)=\textbf{0} \qquad
[\sigma(0)](0)= [\sigma(0)](1, \varphi)=0\,,
\end{array}
\end{equation}
This eigenvalue problem could be  {\em formally} addressed by the same procedure and tools employed  in B\'enard convection between two horizontal  planes (lines, in 2-D). Thus, in this analogy, after the seminal work of Reid and Harris \cite{RH}, one {\em expects} that the eigenfunction corresponding to the least eigenvalue possesses the {\em even symmetry}. If this is rigorously ascertained, then, by means of very well-known perturbative arguments, one could show that the same property continues to hold also for finite ``small" $\varepsilon$, that is, $\mathcal D$. Such a study will be the object of future work.
\medskip\par
{\bf Acknowledgments.} We would like to thank the reviewers for their very useful comments and suggestions that led to a rather improved version of the original manuscript.

\hspace{3mm}

\noindent
Antonino De Martino,
Dipartimento di Matematica \\ Politecnico di Milano\\
Via Bonardi n.~9\\
20133 Milano\\
Italy

\noindent
\emph{email address}: antonino.demartino@polimi.it\\
\vspace{4mm}

\noindent
Arianna Passerini,
Department of Mathematics and Computer Science\\ University of Ferrara\\
Via Machiavelli 30,
40121 Ferrara\\
Italy

\noindent
\emph{email address}: ari@unife.it\\

\end{document}